\documentclass[a4paper,11pt]{amsart}

\usepackage[latin1]{inputenc}
\usepackage{amsmath}
\usepackage{amsthm}
\usepackage{amscd}
\usepackage{amssymb}
\usepackage{amsfonts}
\usepackage{pb-diagram, pb-xy}
\usepackage[all]{xy}
\usepackage{hyperref}
\usepackage{cleveref}
\usepackage{MnSymbol}
\hypersetup{colorlinks=true}
\usepackage{float}

\newcommand{\bbN}{\mathbb{N}}

\newcommand{\F}{\mathbb{F}}

\newcommand{\Sym} {{\rm S}}
\newcommand{\Mat} {{\rm Mat}}

\newtheorem*{thm*}{Theorem}
\newtheorem{thm}{Theorem}[section]

\newtheorem{lem}[thm]{Lemma}

\newtheorem{cor}[thm]{Corollary}
\newtheorem{rmk}[thm]{Remark}

\title{Group codes over fields are asymptotically good}

\author{Martino Borello}
\thanks{M. Borello is with Universit\'e Paris 8, Laboratoire de G\'eom\'etrie, Analyse et Applications,  LAGA, Universit\'e Sorbonne Paris Nord, CNRS, UMR 7539, F-93526, Saint-Denis, France}
\author{Wolfgang Willems}
\thanks{W. Willems is with Otto-von-Guericke Universit\"at, Magdeburg, Germany,
 and Universidad del Norte, Barranquilla, Colombia}
\date{}

\begin{document}

\begin{abstract} Group codes are right or left ideals in a group algebra of a finite group over a finite field.
Following ideas of Bazzi and Mitter on group codes over the binary
field \cite{BM}, we prove that group codes over finite fields of any
characteristic are asymptotically good.
  \end{abstract}

\maketitle


\section{Introduction}

Let $\F$ be a finite field of characteristic $p$ and let $G$ be a
finite group. By a {\it group code} or, more precisely, a {\it
$G$-code} we denote a right or left ideal in the group algebra $\F
G$. Many interesting linear codes are group codes. For example,
cyclic codes of length $n$ are group codes for a cyclic group $C_n$;
Reed-Muller codes  are group codes for an elementary abelian
$p$-group \cite{Berman, Charpin}; the binary extended self-dual
$[24,12,8]$ Golay code is a group code for the symmetric group
$\Sym_4$ on $4$ letters \cite{Bernhardt} and the dihedral group
$D_{24}$ of order $24$ \cite{McLH}. Many best known codes are group
codes as well. For instance, $\F_5(C_6 \times C_6)$ contains a
$[32,28,6]$  and
$\F_5(C_{12} \times C_6)$ a $[72,62,6]$ group code \cite{JLLX}. Both codes improved  earlier examples in Grassl's list \cite{G}.\\

Already in 1965, Assmus, Mattson and Tyrun \cite{AMT} asked the
question whether the class of cyclic codes, i.e., the class of group
codes over cyclic groups, is asymptotically good. The answer is
still open. In \cite{BM}, Bazzi and Mitter proved that the class of
group codes over the binary field is asymptotically good. Using the
trivial fact that by field extensions neither the dimension nor the
minimum distance changes, group codes are asymptotically good in
characteristic $2$. In this note we use the ideas of Bazzi and
Mitter to prove our main result. \\

\noindent
{\bf Theorem.} Group codes over fields are asymptotically good in any characteristic. \\[1ex]

The proof mainly follows the lines of \cite{BM} and does not
distinguish between the prime $p=2$ and $p$ odd for the
characteristic of
the underlying field.\\

For different primes $p\not= q$ let $s_p(q)$ denote the order of $p$
modulo $q$. In order to construct a sequence of particular binary
group algebras over dihedral groups, in \cite{BM} the authors need a
set of primes $q$ with $2 \mid s_2(q)$ which has positive density in
the set of all primes. Such a set is obviously given by all primes
$q \equiv \pm 5 \bmod 8$. For odd primes $p$ the analog is far less
obvious, but has already been proved by Wiertelak in 1977 (see
\cite{Moree}). In the following unified proof (i.e., $p$ any prime)
we heavily use results from modular representation theory.

\section{The structure of the group algebra {$\F_p G_{p,q,m}$}} \label{algebra}

Let $p$ be a fixed prime and let $q$ be a prime such that $p$
divides $q-1$ (there are infinitely many such $q$, by Dirichlet's
Theorem). For $m \in \bbN$ such that $m \not\equiv 1 \bmod q$ and
$m^p\equiv 1\bmod q$, we define the group $G_{p,q,m}$ by
\begin{equation}\label{eq:grp}
G_{p,q,m}:=\langle \alpha,\beta \ | \ \alpha^p=\beta^q=1,
\alpha\beta\alpha^{-1}=\beta^m\rangle=\langle\beta\rangle\rtimes
\langle\alpha\rangle.
\end{equation}
Note that $G_{p,q,m}$ is a nonabelian metacyclic group. In the case
$p=2$ and $m= q-1$ the group $G_{2,q,q-1}$ is a dihedral group
which has been considered in \cite{BM} to prove the Theorem over the binary field $\F_2$.\\

Next we put $N:=\langle\beta\rangle$ and $Q:=\F_pN$. Any element $r$
of $\F_pG_{p,q,m}$ can uniquely be written  as
$$r=r_0+\alpha r_1+\cdots+\alpha^{p-1}r_{p-1}$$
with $r_0,\ldots,r_{p-1} \in Q$. If $a=\sum_{i=0}^{q-1} a_i \beta^i$
(with $a_i\in\F_p$) is an element of $Q$, we  define $\hat{a}$ by
$$\hat{a}:=\sum_{i=0}^{q-1} a_i \beta^{i\cdot m}$$
Clearly, the map  $\hat{}:Q\to Q$ is an $\F_p$-algebra automorphism.
From the relation $\alpha\beta=\beta^m\alpha$ we get
 $\alpha\beta^i=\beta^{i\cdot m}\alpha$ for all
$i\in \{0,\ldots,q-1\}$, so that
$$\alpha a=\hat{a}\alpha$$
for all $a\in Q$.

Now we realize  $Q$ as $\F_p[x]/\langle x^q-1\rangle$.
Since $Q$ is a semisimple algebra by Maschke's Theorem (\cite{AB}, p. 116), we have,  due to Wedderburn's Theorem
(\cite{AB}, Chap. 5, Sect. 13, Theorem 16), a unique decomposition
$$Q=\bigoplus_{i=0}^sQ_i$$
into $2$-sided ideals $Q_i$, where each
$Q_i$ is a simple algebra over $\F_p$. If $$x^q-1=\prod_{i=0}^s
f_i$$ is a factorization of $x^q-1$ into irreducible polynomials $f_i \in
\F_p[x]$, then
$$Q_i=\left\langle \frac{x^q-1}{f_i} \right\rangle\cong \F_p[x]/\langle f_i\rangle\cong \F_{p^{\deg f_i}}.$$
We may  suppose that $f_0=x-1$, so that $Q_0=\langle 1+\ldots+x^{q-1}\rangle\cong \F_p$.\\
Now let $\zeta_q$ be a primitive $q$-th root of unity in an extension field of
$\F_p$. It is well-known by basic Galois theory that, for every
$i\in \{1,\ldots,s\}$, there exists exactly one coset $A_i$ in
$\F_q^\times /\langle p\rangle$ such that
$$f_i=\prod_{a\in A_i} (x-\zeta_q^a)$$
and the map $f_i\mapsto A_i$ is one-to-one. Furthermore, $\deg f_i
=s_p(q)$, which is the multiplicative order of $p$ in $\F_q^\times$.
In particular,
$$\dim Q_i := l_i = s_p(q) $$ for $i\in\{1,\ldots,s\}$.
The automorphism \ $\hat{}$ \  maps each $Q_i$ to some $Q_j$. More
precisely, $\hat{Q_i}$ corresponds to the coset $mA_i$. In
particular, $\hat{Q_i}=Q_i$ iff $mA_i=A_i$. \\ 

In what follows we need to understand  which conditions on $q$ and $m$ imply
$\hat{Q}_i=Q_i$ for all $i\in \{1,\ldots,s\}$. Note that obviously
$\hat{Q}_0 = Q_0$.

\begin{lem} \label{T1} The following
conditions are equivalent.
\begin{enumerate}
    \item $\hat{Q_i}=Q_i$ for all $i\in \{0,1,\ldots,s\}$.
    \item There exists $i\in \{1,\ldots,s\}$ such that $\hat{Q_i}=Q_i$.
    \item $m \in \langle p\rangle \leq \F_q^\times$.
\end{enumerate}
\end{lem}

\begin{proof}
Clearly $(1)$ implies $(2)$. By the  discussion above,
$\hat{Q_i}=Q_i$ for some $i \geq 1$ iff $mA_i=A_i$, which happens
iff $m \in \langle p\rangle \leq \F_q^\times$. So $(2)$ implies
$(3)$. Obviously $(1)$ follows from $(3)$.
\end{proof}

Let $s_p(q)$ denote the order of $p$ modulo $q$ and suppose that $
p\mid s_p(q)$. Thus $s_p(q)= pu$ for some $u\in \bbN$. We may take
$m:=p^u$ in the definition of $G_{p,q,m}$, since $m \not\equiv 1
\bmod q$ and $m^p\equiv 1\bmod q$.
In this case we have $\hat{Q_i} =Q_i$ for $i\in\{0,1,\ldots,s\}$, by Lemma \ref{T1}. \\

Now let
$$ {\mathcal P} :=\{q\mid q \ \text{a prime}, \ p \mid s_p(q) \}.$$
The set ${\mathcal P}$ of primes is
infinite and it has positive density (see for instance \cite{Moree}).\\

\textbf{From now on, we assume that $q \in \mathcal{P}$.}

\bigskip

Let $G:=G_{p,q,p^{s_p(q)/p}}$ and recall that $Q = \F_pN = Q_0
\oplus \ldots \oplus Q_s$ with $Q_0 =
(\sum_{i=0}^{q-1}\beta^i)\F_p$. If we put
$$ R_i= Q_i \oplus \alpha Q_i \oplus \ldots \oplus \alpha^{p-1} Q_i$$
for $i\in \{0,\ldots,s\}$, then obviously
$$ \F_pG = R_0 \oplus \ldots \oplus R_s.$$

\begin{thm} \label{structure1} The structure of $R_i$ is as follows.
\begin{itemize}
\item[a)] All $R_i$ are $2$-sided ideals of $\F_pG$.
\item[b)] As a left $\F_pG$-module we have $R_0 \cong \F_pG/N$. In particular, $R_0$ is uniserial of dimension $p$ and
all composition factors are isomorphic to the trivial
$\F_pG$-module.
\item[c)] For $i>0$ all minimal left ideals in $R_i$ are projective $\F_pG$-modules. Thus $R_i$ is a completely reducible
left $\F_pG$-module for $i>0$.
\item[d)] $R_i$ is indecomposable as a $2$-sided ideal, hence a {\rm $p$-block of $\F_pG$}. In particular,
 $R_i$ contains up to isomorphism exactly one irreducible left $\F_pG$-module
which is of dimension $l_i = s_p(q)$.
\item[e)] $R_i \cong \Mat_p(\F_{{p}^{{l_i}/{p}}})$ for $i>0$ and $R_i$ contains up to isomorphism exactly one irreducible
left $\F_pG$-module, say $M_i$, of dimension $l_i =s_p(q)$.
\end{itemize}
\end{thm}
\begin{proof} a) Clearly, $R_i$ is a left ideal. It is also a right ideal since $Q_i= \hat{Q}_i$ by
Lemma \ref{T1}, and $\alpha a = \hat{a} \alpha$ for $a \in Q$. \\
b) This follows immediately from representation theory (see for instance (\cite{HB}, Chap. VII, Example 14.10)). \\
c) Let $\bar{\F}_p \supseteq \F_p$ be a finite splitting field for
$N$ (\cite{HB}, Chap. VII, Theorem 2.6). Thus every irreducible
character $\chi$ of $\bar{\F}_pN$ is of degree $1$. If $\chi$ is not
the trivial character, then, according to the action of $\alpha$ on
$\beta$,  the induced character $\chi^G$ is an irreducible character
for $G$, by Clifford's Theorem. Furthermore $\chi^G$ is afforded by
an irreducible projective $\bar{\F}_pG$-module (\cite{HB}, Chap.
VII, Theorem 7.17). Thus all non-trivial irreducible
$\bar{\F}_pG$-modules are projective. Now, let $M$ be an irreducible
non-trivial $\bar{\F}_pG$-module and denote by $M_0$ the space $M$
regarded as an $\F_pG$ module. Then, by (\cite{HB}, Chap. VII,
Theorem 1.16 a)), $M_0 \otimes_{\F_p} \bar{\F}_p$  is a direct sum
of Galois conjugates of $M$,
 which are all projective since no one is the trivial module. Finally, by (\cite{HB}, Chap. VII, Ex. 19 in Sec. 7),
the module $M_0$ is a projective $\F_pG$-module, and by (\cite{HB}, Chap. VII, Theorem 1.16 d)),
$M_0 \cong W \oplus \ldots \oplus W$ for some irreducible $\F_pG$-module $W$.
 Thus $W$ is projective. Since obviously all irreducible non-trivial $\F_pG$-modules
can be described this way we are done.\\
d) Note that $R_i$ is not irreducible as a left module since $M_i:=Q_i(1+ \alpha + \ldots + \alpha^{p-1})$ is a minimal ideal in $R_i$.
Clearly, $Q_i \cong M_i$ as a left $\F_pN$-module.
Thus $Q_i$ has an extension  to the irreducible $\F_pG$-module $M_i$.
But all extensions are isomorphic since $G/N$ is a $p$-group. Thus $R_i$ has up to isomorphism exactly one irreducible $F_pG$-module and $\F_pG$
 has exactly $s+1$ non-isomorphic $\F_pG$-modules.
 If some $R_i$ is a direct sum of two non-zero $2$-sided
ideals, then $R_i$ contains at least two non-isomorphic irreducible
$\F_pG$-modules, a contradiction. \\
e) By c) and d), we know that $R_i$ contains up to isomorphism
exactly one irreducible left $\F_pG$-module, say $M_i$, which has
dimension $l_i$.  Thus $R_i \cong M_i \oplus \ldots
\oplus M_i$ with $p$ components $M_i$.
That $R_i$ has the indicated matrix algebra structure now follows by
Wedderburn's Theorem.
\end{proof}

\begin{lem} \label{Z_i} For $i>0$ we have
\begin{itemize}
\item[a)] $Z_i:=\{a\in Q_i \mid a =\hat{a} \}$ is a subfield of $Q_i$.
\item[b)] $\#Z_i = p^\frac{l_i}{p} = p^\frac{s_p(q)}{p}$.
\end{itemize}
\end{lem}
\begin{proof} a) This is obviously true. \\
b) Since $\alpha$ acts fixed point freely on $N\setminus\{1\}$ we
get $\dim \{a \in Q^\ast \mid \hat{a} = a\} = \frac{q-1}{p}$. Now,
it is sufficient to show that $\dim Z_1 = \dim Z_j$ for $j\geq 1$,
which implies
$$ \dim Z_i = \frac{q-1}{sp} = \frac{s_p(q)}{p} = \frac{l_i}{p}.$$
Let $\bar{\F}_p$ be a splitting field for $G$. To prove that $\dim
Z_1 = \dim Z_j$ for $j \geq 1$ first note that $ Q_i \otimes_{\F_p}
\bar{\F}_p = V_1 \oplus \ldots \oplus V_{l_i}$, where $V_j =
(\frac{1}{|N|} \sum_{x \in N} \chi_j(x^{-1})x)\bar{\F}_p$ and
$\chi_j$ is a linear non-trivial character of $\bar{\F}_p N$. Thus
$\alpha$ acts regularly on the set $\{V_1, \ldots, V_{l_i}\}$, which
proves that the fixed point space of $\alpha$ on $V_1 \oplus \ldots
\oplus V_{l_i}$ has dimension $\frac{l_i}{p}$. This implies that the
fixed point space on $W_i$ also has dimension $\frac{l_i}{p}$, i.e.
$\#Z_i = p^\frac{l_i}{p}$.
\end{proof}

In order to determine all minimal left ideals in $R_i$ we need the following notation.
For $b \in Q_i^\times$ we denote by $[b]$ the image of $b$ in the factor group $Q_i^\times/Z_i^\times$.

\begin{lem} \label{min-ideals} For $i>0$ we have the following.
\begin{itemize}
\item[a)] For $b \in Q_i^\times$, the space  $Q_i(1 + \alpha + \ldots \alpha^{p-1})b$  is a minimal left ideal in $R_i$.
\item[b)] $Q_i(1 + \alpha + \ldots \alpha^{p-1})b = Q_i(1 + \alpha + \ldots \alpha^{p-1})b'$ iff $[b] = [b']$.
\item[c)] Each minimal left ideal of $R_i$ is of the form $I_{[b]}^i=Q_i(1 + \alpha + \ldots \alpha^{p-1})b$ with $b \in Q_i^\times$.
\end{itemize}
\end{lem}
\begin{proof} a) This is clear since $\alpha a = \hat{a} \alpha$ for $a \in Q$ and $\hat{Q_i} = Q_i$. \\
b) Suppose that $0 \not= a(1 + \alpha + \ldots \alpha^{p-1})b = a'(1 + \alpha + \ldots \alpha^{p-1})b'$ with $a,a',b,b' \in Q_i^\times$.
Thus $$x(1 + \alpha + \ldots \alpha^{p-1})y=(1 + \alpha + \ldots \alpha^{p-1})$$ with $x = a'^{-1}a$ and $y =bb'^{-1}$.
Since
$$ x(1 + \alpha + \ldots \alpha^{p-1})y = xy + x\hat{y}\alpha  + \hat{\hat{y}}\alpha^2 + \ldots $$
we obtain $xy =1 = x\hat{y}$, hence $y = \hat{y}$. It follows
$$ y = bb'^{-1} \in Z_i^\times, $$
hence $[b] = [b']$. Conversely, if $[b] = [b']$,
then obviously $Q_i(1 + \alpha + \ldots \alpha^{p-1})b = Q_i(1 + \alpha + \ldots \alpha^{p-1})b'$. \\
c) Since $\# Z_i = p^{\frac{l_i}{p}}$ by Lemma \ref{Z_i}, we have constructed so far exactly  $\frac{p^{l_i}-1}{p^{l_i/p}-1}$
minimal left ideals. According to Lemma \ref{structure1} e)  we have $R_i \cong \Mat_p(\F_{{p}^{{l_i}/{p}}})$. It is well-known that
there is a bijection between the set of minimal left ideals in $\Mat_p(\F_{{p}^{{l_i}/{p}}})$ and the set of 1-dimensional subspaces
in a $p$-dimensional vector space over $\F_{{p}^{{l_i}/{p}}}$, which has cardinality $\frac{p^{l_i}-1}{p^{l_i/p}-1}$.
\end{proof}

\section{Asymptotically good group codes} \label{asymp}

In this section we prove that group codes are asymptotically good in
any characteristic. We set here $G:=G_{p,q,p^{s_p(q)/p}}$ and we
consider the group algebra $\F_pG$. All the notations are as in
Section \ref{algebra}.

\begin{lem}[Chepyzhov \cite{Ch}]\label{lem:asymp}  Let $r:\bbN \longrightarrow \bbN$ denote a non-decreasing function and let
$$P(r)=\{t \text{ prime} \ | \ s_p(t)\geq r(t)\}.$$ If $r(t)<<\sqrt{\gamma\cdot t/\log_p t}$, with $\gamma=\log_p(e)\cdot\log_p(2)$,
then $P(r)$ is infinite and dense in the set of all primes. In
particular, if $\log_p(t)<<r(t)<<\sqrt{\gamma\cdot t/\log_p t}$,
then the set of primes $t$ such that $s_p(t)$ grows faster than
$\log_p(t)$ is infinite and dense in the set of all primes.
\end{lem}

\begin{proof}
Let $B_n$ be the set of primes $t$ less than $n$ which are not in
$P(r)$ (i.e., if $\pi(n)$ is the set of primes less than $n$, then
$\pi(n)=B_n\sqcup (P(r)\cap \pi(n))$). Since $s_p(t)$ is the
multiplicative order of $p$ modulo $t$, there exists, for every $t$
in $B_n$, two integers $a \in \bbN$ and $k\in \bbN$ such that
$$0<a<r(t) \text{ and } p^a-1=kt.$$
Thus
$$\#B_n\leq \#\{(a,k) \ | \ 0<a<r(t) \text{ and } (p^a-1)/k \text{ is prime}\}\leq r(t)\cdot\max_{0<a<r(t)}\#\{\text{prime factors of }p^a-1\}$$
$$\leq r(t)\cdot \log_2(p^{r(t)}-1)\leq r(t)^2\cdot \log_2(p)<< \frac{t}{\log t}.$$
By the Prime Number Density Theorem, we have $\pi(n)\sim n/\log n$. Thus the
set $P(r)$ is infinite, even dense in the set of all primes.
\end{proof}

\begin{rmk}\label{rmk:infinitelymany}
Since $\mathcal{P}$ has positive density, there are infinitely many
$q\in \mathcal{P}$ such that $s_p(q)$ grows faster than $\log_p(q)$.
\end{rmk}

\begin{lem}\label{lem:Omega}
If $\Omega_l$ be the set of left ideals in $Q$ of dimension $l$,
then $\#\Omega_l\leq q^{l/s_p(q)+1}$.
\end{lem}

\begin{proof}
Recall that $Q_0,Q_1,\ldots,Q_s$ are the irreducible modules in $Q$
where $\dim_{\mathbb{F}_p}Q_0=1$ and $\dim_{\mathbb{F}_p}Q_i=s_p(q)$
for $i\in \{1,\ldots,s\}$. An ideal of dimension $l$ is a direct sum
of at most $l/s_p(q)+1$ of these irreducible modules. There are at
most $(s+1)^{l/s_p(q)+1}$ such sums and the assertion follows from
$s+1\leq q=s_p(q)\cdot s+1$.
\end{proof}

Let $Q^\ast=\bigoplus_{i=1}^s Q_i$ and let $Q^{\ast \times}$ be the
multiplicative group of units of $Q^\ast$.

\begin{lem}\label{lem:U}
If $f\in Q^\ast$ such that $\dim fQ=l$ and
$$U=Q^{\ast \times}f(1+\alpha+\ldots+\alpha^{p-1})Q^{\ast \times},$$
then $\#U\geq p^{{\frac{2p-1}{p}}l}$.
\end{lem}

\begin{proof}
We may decompose $f=\sum_{i=1}^s f_i$, with $f_i\in Q_i$ and put
$S:=\{i \ | \ f_i\neq 0\}$. Since $f_iQ_i^\times=Q_i^\times$ for
$i\in S$ (recall that $Q_i$ is isomorphic to a field), we get
$$U=\sum_{i\in S} Q_i^{\times}(1+\alpha+\ldots+\alpha^{p-1})Q_i^{\times}.$$
By Lemma \ref{min-ideals}, we have
$$Q_i^{\times}(1+\alpha+\ldots+\alpha^{p-1})Q_i^{\times}=\bigsqcup_{[b]\in Q_i^\times/Z_i^\times} I^i_{[b]}\setminus \{0\},$$
where $  \# I^i_{[b]} =p^{l_i}$ and
$\#Q_i^\times/Z_i^\times=\#\{\text{irreducible left ideals in } R_i\}=\frac{p^{l_i}-1}{p^{l_i/p}-1}.$
It follows
$$\#(Q_i^{\times}(1+\alpha+\ldots+\alpha^{p-1})Q_i^{\times})=\frac{p^{l_i}-1}{p^{l_i/p}-1}\cdot (p^{l_i}-1)\geq  p^{(p-1)l_i/p}\cdot p^{l_i}.$$
Finally,
$$\#U\leq \sum_{i\in S}p^{(p-1)l_i/p}\cdot p^{l_i}= p^{{\frac{2p-1}{p}}l},$$
since $l=\sum_{i\in S}l_i$.
\end{proof}

In order to prove Theorem \ref{T2} we need the following result
which is a special case of (\cite{FL}, Theorem 3.3). Let us recall that a group code is a balanced code, as observed in \cite[Lemma 2.2.]{BM}.

\begin{lem} \label{Fan-Liu} Let $C$ be a $[n,k]_p$ group code.  Then
$$A_w(C) := \#\{c\mid c \in C, \, {\rm wt}(c) = w \}\leq p^{k\cdot h_p(w/n)}$$
for all $0\leq w\leq \frac{p-1}{p}\cdot n$, where
$$h_p(x):=-(1-x)\log_p(1-x)-x\log_p\left(\frac{x}{p-1}\right)$$ is the
$p$-ary entropy function.
\end{lem}

\begin{thm} \label{T2}
Let $R:=\F_pG$ and consider the unique decomposition
$R=\bigoplus_{i=0}^sR_i$ into the $p$-blocks $R_i$ described in
Theorem {\rm \ref{structure1}}.

Now we choose a left ideal $I$ of $R$ as
$$I=\bigoplus_{i=1}^sI_i$$
where each $I_i$ is taken uniformly at random among the $1+
p^{l_i/p}+ \ldots + p^{(p-1)l_i/p}$ non-zero irreducible left ideals
of $R_i$.\\

If $0<\delta\leq \frac{p-1}{p}$ satisfies $h_p(\delta)\leq \frac{p-1}{p^2}-\frac{\log_p(q)}{p\cdot s_p(q)}$,
then the probability that the minimum relative distance of $I$ is below $\delta$ is at most
 $$p^{-p\cdot s_p(q)\cdot \left(\frac{p-1}{p^2}-h_p(\delta)\right)+(2p+1)\log_p(q) }. $$ \\
\end{thm}

\begin{proof}
Since every irreducible left ideal $I_i$ is of the form given in Lemma \ref{min-ideals},
 the above randomized construction is equivalent to consider
$$I_{[b]}=Q(1+\alpha+\ldots+\alpha^{p-1})b=Q^\ast(1+\alpha+\ldots+\alpha^{p-1})b$$
where $[b]$ is selected uniformly at random from $Q^{\ast \times}/Z$
with $Z:=\{a \in Q^{\ast \times} \mid \hat{a}= a\}$. Since $Q^{\ast
\times}$ is a group, we have $Q^{\ast \times}=aQ^{\ast \times}$ for
all $a\in Q^{\ast \times}$, hence
$$I_{[b]}=aQ^\ast(1+\alpha+\ldots+\alpha^{p-1})b$$
for  all $a\in Q^{\ast \times}$. Let
$$P={\rm Pr}({\rm d}(I_{[b]})\leq pq\delta)=\frac{\#\{I_{[b]} \ | \ {\rm d}(I_{[b]})\leq pq\delta\}}{\#(Q^{\ast \times}/T)}=\frac{\#\{(a,b) \ | \ {\rm d}(aQ^\ast(1+\alpha+\ldots+\alpha^{p-1})b)\leq pq\delta\}}{\#(Q^{\ast \times})^2}.$$
By definition of the minimum distance, we have that
$$P\leq \sum_{f\in Q^\ast,f\neq 0}{\rm Pr}_{(a,b)\in (Q^{\ast \times})^2}(0\leq {\rm wt}(af(1+\alpha+\ldots+\alpha^{p-1})b)<pq\delta).$$
We can partition $Q$ as
$$Q=\bigsqcup_{l=s_p(q)}^q \underbrace{\{f\in Q \ | \ \dim_{\mathbb{F}_p} fQ=l\}}_{=D_l} \ \text{ and } \ Q^\ast=\bigsqcup_{l=s_p(q)}^q \underbrace{D_l\cap Q^\ast}_{=D_l^\ast},$$
so that
$$P\leq \sum_{l=s_p(q)}^q \#(D^\ast_l)\max_{f\in D^\ast_l} {\rm Pr}_{(a,b)\in (Q^{\ast \times})^2}(0\leq {\rm wt}(af(1+\alpha+\ldots+\alpha^{p-1})b)<pq\delta).$$
Let $\Omega_l$ be the set of left ideals in $Q$ of dimension $l$.
Then
$$\#(D^\ast_l)\leq \#(D_l)\leq p^l\cdot \#(\Omega_l)\leq p^l\cdot q^{l/s_p(q)+1}$$
by Lemma \ref{lem:Omega}. For any $l$ and any $f\in D_l^\ast$, we
can define
$$U=Q^{\ast \times}f(1+\alpha+\ldots+\alpha^{p-1})Q^{\ast \times}$$
as in Lemma \ref{lem:U}.  Using this we get
$${\rm Pr}_{(a,b)\in (Q^{\ast \times})^2}(0\leq {\rm wt}(af(1+\alpha+\ldots+\alpha^{p-1})b)<pq\delta)=$$ $$=\sum_{r\in U, 0\leq {\rm wt}(r)<pq\delta} {\rm Pr}_{(a,b)\in (Q^{\ast \times})^2}(af(1+\alpha+\ldots+\alpha^{p-1})b=r)\leq$$
$$\leq \max_{r\in U} {\rm Pr}_{(a,b)\in (Q^{\ast \times})^2}(af(1+\alpha+\ldots+\alpha^{p-1})b=r) \ \cdot$$ $$\cdot \ \sum_{w_1,\ldots,w_p\geq 0, w_1+\ldots+w_p<pq\delta}\#(fQ^{(w_1)})\cdot\ldots \cdot \#(fQ^{(w_p)}),$$
where $fQ^{(w)}$ is the set of elements of weight $w$ in $fQ$.\\
It is easy to see that each $r\in U$ can occur with the same
probability as $af(1+\alpha+\ldots+\alpha^{p-1})b$, so that the
above probability is independent of $r$. Thus we have
$${\rm Pr}_{(a,b)\in (Q^{\ast \times})^2}(af(1+\alpha+\ldots+\alpha^{p-1})b=r)=\frac{1}{\#U}\leq  p^{-{\frac{2p-1}{p}}l},$$
by Lemma \ref{lem:U}.\\
Moreover, $fQ$ is a $[pq,l]_p$ group code,
so that, by Lemma \ref{Fan-Liu}, we have
$$\#(fQ^{(w)})\leq p^{l\cdot h_p(w/pq)}$$
for all $w\leq (p-1)\cdot q$ (which is true, since $\delta\leq
\frac{p-1}{p}$). Putting together all previous inequalities we have
$$P\leq \sum_{l=s_p(q)}^q p^{-\frac{p-1}{p}l}\cdot q^{l/s_p(q)+1}\cdot  \sum_{w_1,\ldots,w_p\geq 0, w_1+
\ldots+w_p<pq\delta}p^{l\cdot \sum_{i=1}^ph_p(w_i/pq)},$$
so that, by the convexity,
$$P\leq \sum_{l=s_p(q)}^q p^{-\frac{p-1}{p}l}\cdot q^{l/s_p(q)+1}\cdot (pq\delta)^p\cdot p^{l\cdot p\cdot h_p(\delta)}\leq
\sum_{l=s_p(q)}^q p^{l\cdot p\cdot \left(h_p(\delta)-\frac{p-1}{p^2}+\frac{\log_p(q)}{p\cdot s_p(q)}\right)+p+p\log_p(q)}.$$
Finally, if $h_p(\delta)\leq \frac{p-1}{p^2}-\frac{\log_p(q)}{p\cdot
s_p(q)}$, then
$$P\leq p^{-p\cdot s_p(q)\cdot \left(\frac{p-1}{p^2}-h_p(\delta)\right)+(p+1)\log_p(q)+p}
\leq p^{-p\cdot s_p(q)\cdot \left(\frac{p-1}{p^2}-h_p(\delta)\right)+(2p+1)\log_p(q)}.
$$
\end{proof}

\begin{cor} \label{cor} Group codes over finite fields are asymptotically good.
\end{cor}
\begin{proof} We have to prove the assertion only for prime fields. The general case then  follows by field extension
(see (\cite{FW}, Proposition 12)). According to Lemma
\ref{lem:asymp} and Remark \ref{rmk:infinitelymany}, we may choose a
sequence of primes $q_i$ in ${\mathcal P}$ such that $q_1< q_2 <
\ldots $ and $\frac{s_{p}(q_i)}{\log_p(q_i)} \longrightarrow \infty$ for
$ i \longrightarrow \infty$. Let $0<\delta\leq \frac{p-1}{p}$ with
$h_p(\delta)\leq \frac{p-1}{p^2}-\frac{\log_p(q_1)}{p\cdot
s_p(q_1)}$. Thus the assumption in Theorem \ref{T2} is satisfied for
all $q_i$ and we can find a left ideal $I_{q_i}$ in
$\F_pG_{p,q_i,p^{s_p(q_i)/p}}$ with relative minimum distance at
least $\delta$. Furthermore, $\dim I_{q_i} = s \cdot s_p(q_i) =
q_i-1$. Thus
$$\frac{\dim I_{q_i}}{pq_i} = \frac{1}{p} - \frac{1}{pq_i} \geq \frac{1}{p} - \frac{1}{pq_1}.$$
This shows that the sequence of the left ideals $I_{q_i}$ is asymptotically good.
\end{proof}

\begin{rmk} {\rm  Note that the groups $G_{p,q,m}$ are $p$-nilpotent with cyclic Sylow $p$-subgroups. Thus the asymptotically good sequence
we constructed in Corollary \ref{cor} is a sequence of group codes
in code-checkable group algebras \cite{BCW}.  In such algebras all
left and right ideals are principal.}
\end{rmk}

\noindent {\bf Acknowledgement.} The first author was partially
supported by PEPS - Jeunes Chercheur-e-s - 2018. We are very
grateful to Pieter Moree who brought to our
attention his paper \cite{Moree}. Thanks also goes to anonymous referees who pointed out some inconsistencies in an earlier version.

\end{document}